\documentclass[a4paper]{article}
\usepackage[latin1]{inputenc} 
\usepackage[T1]{fontenc} 
\usepackage{RRA4,RRthemes}
\usepackage{hyperref}
\usepackage{amsmath,amsthm,amsfonts}
\newtheorem{theorem}{Theorem}
\newtheorem{corollary}{Corollary}
\usepackage{epstopdf}
\usepackage[all]{xy}

\usepackage[english]{babel}
\RRdate{April 2011}

\RRauthor{
Fabien Mathieu
}
\authorhead{Fabien Mathieu}
\RRtitle{Utilisation de pairs auxiliaires\\
 pour de la diffusion pair-à-pair en direct}
\RRetitle{On Using Seeders\\ for P2P Live Streaming}
\titlehead{Live Seeding}
\RRresume{
Les \emph{semeurs} (pairs possédant déjà un contenu donné et contribuant à sa dissémination) sont un concept clé du pair-à-pair (P2P). Ils permettent entre autres d'accroître les performances d'un système P2P. Mais alors qu'il est naturel d'avoir des semeurs dans le contexte du partage de fichiers ou de la vidéo-à-la-demande, cela semble incompatible avec de la diffusion en direct. Le but de ce rapport est de montrer que, dans une certaine mesure, cela est réalisable.

Après avoir défini formellement le concept de semeur pour la diffusion en direct, et proposé une définition d'efficacité, nous regardons la performance théorique des semeurs pour des systèmes où le coût de contrôle est négligé. Nous proposons ensuite un modèle avec coût de contrôle affine, et donnons les résultats pour un semeur unique tout comme pour un ensemble de semeurs (un ensemble de semeurs ne se comporte pas nécessairement aussi bien que la somme de ses éléments).
}

\RRabstract{Seeders (peers that do not request anything but contribute to the system) are a powerful concept in peer-to-peer (P2P). They allow to leverage the capacities of a P2P system. While seeding is a natural idea for filesharing or video-on-demand applications, it seems somehow counter-intuitive in the context of live streaming. This paper aims at describing the feasibility and performance of P2P live seeding.

After a formal definition of ``live seeding'' and efficiency, we consider the theoretical performance of systems where the overhead is neglected. We then propose a linear overhead model and extend the results for this model, for a single seeder and for a set of seeders as well (it is not always possible to perfectly aggregate individual efficiencies in a given system).
}
\RRmotcle{pair-à-pair, diffusion en direct, bande passante, semeurs, efficacité}
\RRkeyword{peer-to-peer, live streaming, bandwidth, seeders, performance}
\RRprojet{GANG}  
\RRdomaine{1} 
\RRthemeProj{gang} 
\RRdomaineProjBis{rap} 
\RCParis 

\RRNo{7608}

\begin{document}
\makeRR   

\section{Introduction}
\label{sec:in}

Upload bandwidth is one of the main bottleneck in peer-to-peer (P2P) content distribution, which relies on the upload capacity of its participants to achieve its purpose. The upload resource is all the more critical since most todays high speed Internet access are asymmetric DSLs connections that are not designed to handle P2P traffic and offer relatively low upload capacity, with typical uplink/downlink ratios between $1/4$ and $1/20$. On the one hand, the democratization of very high speed, symmetric, Internet access like FTTH is expected to improve the upload capacity of P2P systems, but on the other hand the evolution of content quality standards makes the requirements in terms of content size and rate higher and higher: earlier video feeds on the Internet where low quality, requiring streamrates of a few hundred kbps, whereas HDTV implies rate of up to 20 Mbps, possibly more with the upcoming of 3D video content. It is therefore likely possible that upload will still be a major bottleneck of tomorrow's P2P content distribution.

\subsection{Motivation}

In order to increase the available resources, a standard P2P technique is to leverage the capacity of the system by using \emph{seeders}, \emph{i.e.} peers that contribute to the system but are (currently) not needing anything. Using seeders is quite natural for file-sharing or Video-on-Demand: after a peer has downloaded its file or video, it becomes a potential seeder for that content. However, it is counter-intuitive live streaming systems: ``live'' content is created on the fly, so it cannot be pro-actively possessed by peers. Therefore, for a peer to act as a seeder, it has to receive at least a part of the corresponding content, which it does not want to watch by definition.

\subsection{Scope and contribution}

The goal of this paper is to describe the feasibility and performance one can expect from P2P live seeding from a bandwidth budget perspective. This generic theoretical framework can be used to derive simple dimensioning rules and recommendations for the design of P2P live streaming with seeders.

In details, we analyze the seeders' efficiency, which is the useful throughput (goodput) they add to the system, compared to their upload capacity. We provide explicit, tight, upper bounds for efficiency, taking the overhead explicitly into account. We also address the aggregation issues that come from using several seeders. We give conditions and simple diffusion schemes that allow to nearly achieve the theoretical bounds, and provide a few simple examples that illustrate the potential of our findings.

%

\paragraph*{Remark} focusing on a single scenario (live streaming) and a single type of peer (seeders) was a deliberate choice, in order to get a clean framework for investigating theoretical performance, especially with regards to the overhead modeling aspects. This does not preclude of possible extensions of the approach presented here to other use cases.

\subsection{Roadmap}

The next Section introduces the models that we use to derive our results. The related work with respect to P2P bandwidth dimensioning is briefly exposed in Section~\ref{sec:related}. In Section \ref{sec:def_eff}, a formal definition of seeders' efficiency is proposed. Section~\ref{sec:optimal} proposes a preliminary study of efficiency for two overhead-free models. This study is a starting point for the main results of this paper, which derive the efficiency of seeders in a model with explicit overhead (Section~\ref{subsec:optimal}). The validity conditions and applications of the results are discussed in Section~\ref{sec:discussion}. Section \ref{sec:conclusion} concludes.

\section{Model}
\label{sec:model}

We consider a live content that needs to be streamed to a set of users at a constant rate $r$. The delivery is handled by a P2P live streaming system.
%
The specificity of live streaming is that the content cannot be prefetched. A play-out buffer may tolerate some jitter, but the live constraints usually limit the size of that buffer to less than a few seconds, so a conservative, yet realistic assumption is that content must received at exactly the rate $r$ during the whole watching experience.
To compare with,
filesharing usually requires no minimal rate, while in the case of Video-on-Demand, content may be prefetched at a rate greater than $r$.



\subsection{$C/S/L$ systems}


We classify the nodes of the system into three categories:
\begin{itemize}
    \item \emph{Central servers} are in charge of injecting initial copies of the stream into the system.
    We assume they have a cumulated bandwidth capacity that allows to inject $N_C$ copies of the stream, with $N_C\geq 1$.
    \item \emph{Leechers} are peers that want to watch the live content.
    \item \emph{Seeders}\footnote{\label{foot:leech} The terms \emph{leecher} and \emph{seeder} comes from the BitTorrent vocabulary~\cite{cohen03incentives}.}  are peers that do not want to watch the live content, but can provide bandwidth to the system.
\end{itemize}

\paragraph*{Remark} we do not focus on the way seeders could be enforced in a real live streaming system. However, most of the ideas from P2P filesharing or VoD systems should apply to P2P live streaming. For instance:
\begin{itemize}
	\item Some peers may remain connected to the system even when idle.
	\item In a multi-channel system, leechers from an overprovisioned channel may act as seeders for another channel that lacks resources.
	\item A share-ratio policy can encourage the peers to seed: peers that do not offer enough instant bandwidth may have to act as seeders for a while in order to ``pay'' their bandwidth debt. That kind of policy can be enforced through penalties (no service guarantee, reduced catalog) and rewards (higher QoS, access to premium content).
	\item In the case of networks managed by some ISP or content provider, managed seeders may be deployed by the provider to enhance the system performance.
\end{itemize}


We denote by $C$, $L$ and $S$ the sets of servers, leechers and seeders respectively. The number of leechers (resp. seeders) is denoted by  $N_L$ (resp. $N_S$).
Every peer $p$ in $L$ or $S$ has an upload capacity $u_p\geq 0$ devoted to the service. We assume that the download capacity is always sufficient to support the content rate $r$ and a possible overhead. $U_X$ and $\bar{u}_X$ are respectively the total and average upload bandwidths of set $X$ ($\bar{u}_X=\frac{U_X}{N_X}$). 

Note that the bandwidth distribution of the seeders may differ from the one of the leechers. For instance, if seeders are former leechers forced to remain because of some share-ratio policy, low bandwidth peers will have to seed longer \cite{benbadis08playing}, so the average seeders' bandwidth will be lower than the leecher's one. One the other hand, seeders deployed by some content provider should probably have higher bandwidths.


A \emph{diffusion scheme} for the system is a policy that describes how the content is distributed. We assume here static diffusion schemes: between any two peers (or servers) $p$ and $q$, the scheme gives a stream of goodput $0\leq r_{p,q}\leq r$ that is sent from $p$ to $q$. If $0<r_{p,q}< r$, $r_{p,q}$ is called a substream. For convenience, we consider that the substreams received by a given peer are non-overlapping, so a peer $p$ receives an input of rate
\begin{equation}
i_p=\sum_{q \in \{L,S,C\}}r_{q,p}\text{.}
\label{eq:input_sum}
\end{equation}

\paragraph*{Remark} overlapping substreams can always be seen as non-overlapping ones: if $r_{p,q}$ and $r_{s,q}$ are overlapping, with redundant data of rate $r_{p\cap s,q}$, we just have to consider $\tilde{r}_{p,q}:=r_{p,q}-r_{p\cap s,q}$ and see a rate $r_{p\cap s,q}$ from $p$ to $q$ as overhead. Of course, choosing which redundant data is treated as overhead is arbitrary.

Servers apart, a node cannot send something it doesn't possess, so a diffusion scheme verifies the condition
\begin{equation}
\forall p,q\in \{L,S\}, r_{p,q}\leq i_p\text{.}
\label{eq:input_feasibility}
\end{equation}

A scheme is a \emph{solution} of the live diffusion if it ensures that all leechers can view the content, i.e.
\begin{equation}
\forall p\in \{L\}, i_p=r\text{.}
\label{eq:def_solution}
\end{equation}

%
%
%
%

\subsection{Connectivity}

In this work, we use an explicit linear overhead to account for connectivity constraints. We also propose two simpler models that will serve for didactic purposes: perfect systems and limited fanout systems.


\subsubsection{Perfect systems}

In perfect systems, peers can arbitrarily use the upload capacity devoted to the service at no cost~\cite{liu08performance}. In particular, a perfect system possesses the following properties:
\begin{itemize}
	\item \emph{No overhead}: all the bandwidth capacity can be used to effective data transfer (goodput);
	\item \emph{Unlimited fanout}: one peer can send data to an arbitrary numbers of other peers simultaneously;
	\item \emph{Stream continuity}: the live stream can be divided into arbitrary small substreams of constant rate.
\end{itemize}

\subsubsection{Limited fanout}

As we will see in Section \ref{sec:optimal}, optimizing perfect systems often leads to full mesh solutions, which are not very practical. A first idea to make the model more realistic, without explicitly considering the overhead, is to assume that the number of non-null substreams $r_{p,q}$ is limited: each peer $p$ has a limit $c_p$ on the number of outgoing connections it can sustain. This limited fanout implicitly acknowledges the fact that managing a connection has a cost. Perfect systems correspond to the extreme case $c_p=\infty$, for all $p\in \{L,S\}$).


\subsubsection{Explicit linear overhead}

In order to get a more realistic and flexible model of real systems, we propose to assume that the overhead is linear: the actual bandwidth used for sending some content at rate $e$ from one peer to another is $(1+a)e+b$, for some constants $a,b\geq0$. $a$ is the \emph{proportional cost} and $b$ the \emph{additive cost}. For simplicity, we consider that the overhead cost is supported by the sender only (this assumption will be discussed in \ref{sssec:receiver_overhead}). 

The motivation for this model is that most existing sources of overhead are, at least in a rough approximation, proportional or additive:
\begin{itemize}
	\item Periodic signaling messages (keep-alive, overlay maintenance) are additive;
	\item In chunk-based systems, the stream is split into atomic units of data (the chunks) that are distributed independently. For a constant chunk size, the signaling for sending one chunk is expected to induce a proportional overhead;
	\item The cost for initiating a connection, averaged over the lifetime of that connection, can be considered as additive;
	\item Some randomized diffusion scheme can have a non-null probability to to send useless data, because it is outdated or redundant~\cite{bonald08epidemic}. This can be considered as proportional overhead.
\end{itemize}

Under the linear overhead model, a peer of bandwidth $u$ maintaining $c$ outgoing connections has a useful output limited to $\frac{u-bc}{1+a}$. For $b>0$, $\lfloor\frac{u}{b}\rfloor$ is the maximal fanout sustainable by that peer.
 For $b=0$, the model is indeed equivalent to perfect systems, except that all bandwidth capacities have to be normalized by $\frac{1}{1+a}$.




The notation used is summarized in Table \ref{tab:notation}.

\begin{table}[h]
	\centering
\caption{Table of notation}
		\label{tab:notation}
		\begin{tabular}{|c|l|}
		\hline
		$r$ & Streamrate of the content (constant) \\
		$u_p$ & Available upload bandwidth of peer $p$ \\
			$U_X/\bar{u}_X$ & Total/average upload capacity of population $X$ \\
			\hline
  		$N_X$ & Number of nodes in $X$ \\
  		$N_C$ & Normalized capacity of servers ($U_C=N_CR$) \\
  		$i_p$ & Input rate of node $p$ \\
  		$r_{p,q}$ & Substream from $p$ to $q$ \\
  		$\eta_d(X)$ & Efficiency of set $X$ in diffusion scheme $d$\\
  		$c_p$ & Fanout of peer $p$\\
  		$a$ & Proportional cost of a connection\\
  		$b$ & Constant cost of a connection\\
  		$R:=(1+a)r+b$ & Bandwidth consummed by goodput $r$\\
			\hline
		\end{tabular}
\end{table}

\section{Related work}
\label{sec:related}

Understanding the bandwidth dimensioning is a crucial question in P2P systems, as upload bandwidth is a scarce resource. The bandwidth conservation law~\cite{benbadis08playing} tells that, if all available bandwidth resources can be used to useful content transfer, then the condition for a live streaming system to admit a solution is
\begin{equation}
\alpha_L+\beta\alpha_S+\frac{N_C}{N_L}\geq 1\\ \text{, with }\left\{
\begin{array}{l}
	\alpha_X=\frac{\bar{u}_X}{r}\text{,}\\
	\beta=\frac{N_S}{N_L}\text{.}
\end{array}
\right.
\label{eq:bclperfect}
\end{equation}

In reality, not all bandwidth can be used all the time. Of course, there is the issue of overhead, but other phenomena can prevent from using all available bandwidth. For instance, a peer may have nothing to give at a given time; or some bandwidth may be required for other purposes than feeding the leechers. This explains the concept of efficiency. Taking efficiency into account, equation \eqref{eq:bclperfect} becomes
\begin{equation}
\begin{array}{l}
\eta(L)\alpha_L+\eta(S)\beta\alpha_S+\eta(C)\frac{N_C}{N_L}\geq 1 \text{,}\\
\text{where $\eta(X)$ is the efficiency of set $X$.}
\end{array}
\label{eq:bcleta}
\end{equation}

Efficiency was introduced by Qiu and Srikant~\cite{qiu04modeling} for BitTorrent-like file-sharing systems~\cite{cohen03incentives}.
Its role was to quantify the fact that leechers cannot always upload at full bandwidth capacity, as they may lack the content required by others.

In the case of standard peer-assisted live streaming, with no seeders ($S=\emptyset$), Liu \emph{et al.} have shown that one can reach $\eta(L)=\eta(C)=1$ for perfect and limited fanout systems. In other words, a perfect use of the available bandwidth can be achieved~\cite{liu08performance}.







\section{Defining seeders' efficiencies}
\label{sec:def_eff}

We propose to extend the concept of efficiency to seeders as follows: in a given diffusion scheme $d$, the efficiency $\eta_d(s)$ of a seeder $s$ is the ratio between the data bandwidth it adds to the system and its upload bandwidth $u_s$. In the bandwidth budget, we need to acknowledge that the input rate $i_s$ received by $s$  is ``wasted'' : the rate $i_s$ could have been directly sent to some leechers, but instead it is sent to peer $s$, which does not want to watch the content. We say that $s$ ``removes'' $i_s$ from the pool of useful resources, in a matter of speaking\footnote{In fact, deciding whose peer is responsible for the ``waste'' of $i_s$ is arbitrary, and one could decide to substract $i_s$ from the bandwidth of the senders. However, making the seeders responsible for their own input rates make the analysis simpler.}.
So if in $d$, $s$ transmits at rates $r_{s,p_1}$, \ldots, $r_{s,p_c}$ to $c$ other peers (Figure \ref{fig:efficiency_principle}), its efficiency is

\begin{equation}
\eta_d(s):=\frac{\sum_{k=1}^cr_{s,p_k}-i_s}{u_s}\text{.}
\label{eq:efficiency_general}
\end{equation}

\begin{figure}[htb]
\begin{center}
{\tiny
$$\xymatrix{
& & & *++[o][F]{\txt{Seeder/\\Leecher}}\\
& & & *++[o][F]{\txt{Seeder/\\Leecher}}\\
*++[o][F]{\txt{Seeder/\\Leecher/\\Server}} \ar[r]^{i_s} &
*++[o][F]{\txt{Seeder $s$}} \ar[rruu]^{r_{s,p_1}} \ar[rru]_{r_{s,p_2}} \ar[rr]^{r_{s,p_3}} \ar[rrdd]^{r_{s,p_c}}
& & *++[o][F]{\txt{Seeder/\\Leecher}} \ar@{--}[dd]\\
& & & \\
& & & *++[o][F]{\txt{Seeder/\\Leecher}}\\
 } $$
 }
\caption{Principle of live seeding}
\label{fig:efficiency_principle}
\end{center}
\end{figure}
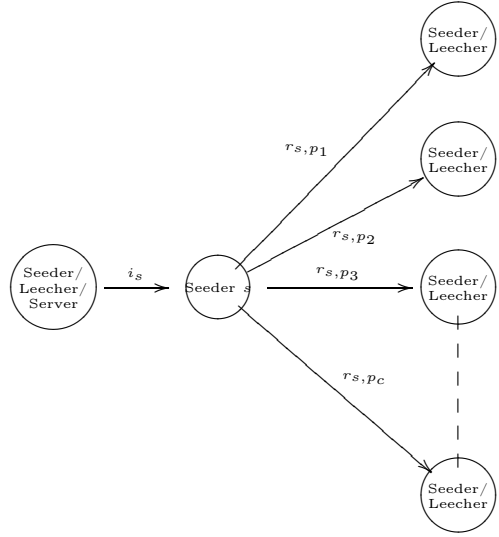

The efficiency of a set $X\subseteq S$ is defined the same way: we consider the difference between what comes out of $X$ and what enters, all reported to capacity:

\begin{equation}
\eta_d(X):=\frac{\sum_{s\in X,q\in\{L,S\setminus X\}}r_{s,q}-\sum_{p\in \{C,L,S\setminus X\},s\in X}r_{p,s}}{U_X}\text{.}
\label{eq:efficiency_general_set1}
\end{equation}

If we add and subtract the term $\sum_{s,t\in X}r_{s,t}$ in the numerator in \eqref{eq:efficiency_general_set1}, we obtain a more compact expression for $\eta_d$:

\begin{equation}
\eta_d(X)=\frac{\sum_{s\in X}\eta_d(s) u_s}{U_X}\text{.}
\label{eq:efficiency_general_set}
\end{equation}

Equation \eqref{eq:efficiency_general_set} tells that $\eta_d(X)$ is also the weighted average of the seeders individual efficiencies.



\subsection{Optimal efficiency}
The optimal efficiency $\eta_{OPT}(s)$ of a seeder $s$ in a given system is defined as the supremum of the efficiencies it can get over all possible diffusion schemes. 
\begin{equation}
\eta_{OPT}(s)=\sup_d(\eta_d(s))
\end{equation}
$\eta_{OPT}(s)$ is an upper bound for the proportion of the upload bandwidth that can be useful for that system.

The same definition stands for the optimal efficiency of any subset $X\subseteq S$:
\begin{equation}
\eta_{OPT}(X)=\sup_d(\eta_d(X))
\end{equation}
However, there is no guarantee that the individual optimal efficiencies of seeders can be aggregated, because they may correspond to distinct schemes (a counter-example is given in Section~\ref{sec:optimal}). As a consequence, Equation \eqref{eq:efficiency_general_set} becomes an inequality when considering optimal efficiency:
\begin{equation}
\eta_{OPT}(X)\leq\frac{\sum_{s\in X}\eta_{OPT}(s)u_s}{U_X}\text{.}
\label{eq:efficiency_opt_def}
\end{equation}



For convenience, subscripts may be omitted when there is no ambiguity. We may also use metonymic notation in order not to clutter notation: $\eta(y)$ may denote the efficiency of a seeder characterised by some property $y$ (like the input rate, upload bandwidth, fanout, \ldots).

\section{Perfect and limited fanout systems}
\label{sec:optimal}


In this section, we derive the optimal efficiency of seeders when there is no explicit overhead.

\subsection{Perfect systems}
\label{sec:ideal}

The optimal performance of seeders in a perfect system is given by the following theorem:
\begin{theorem}
\label{thm:perfect}
The optimal efficiency of a subset $X\subseteq S$ of seeders is
\begin{equation}
\eta(X)=(1-\frac{1}{N_L})\min(1,\frac{N_L r}{U_X})\text{.}
\label{eq:perf_opt}
\end{equation}
\end{theorem}

%

\begin{proof}
First we give a scheme that achieves the efficiency given by \eqref{eq:perf_opt}.
The scheme is the following: each seeder $s\in X$ receives from the servers a distinct substream of rate $\frac{u_s}{N_L}$ (if $U_X\leq N_L r$) or $\frac{u_s}{U_S}r$ (otherwise), and broadcasts that substream to the $N_L$ leechers. Under that scheme, the input received by $X$ from nodes outside $X$ is $\min(\frac{U_X}{N_L},r)$, and the output given to leechers is $\min(U_X,N_L r)$. Subtracting the input from the output and dividing by $U_X$ gives the efficiency $\eta(X)$ from \eqref{eq:perf_opt}.


Then, we need to prove that $\eta(X)$ cannot be greater than $(1-\frac{1}{N_L})\min(1,\frac{N_L r}{U_X})$. If $I_X$ is the input received by $X$ in a given scheme, the corresponding useful output cannot be more that $\min(U_X,\min(I_X,r)N_L)$ because :
\begin{itemize}
	\item $U_X$ is the capacity of $X$;
	\item $\min(I,r)$ is the maximal rate of information that $X$ can get. The best it can achieve is to send that rate to the $N_L$ leechers: sending it to more peers, for instance seeders outside $X$, would be ineffective because all leechers already get the information received by $X$.
\end{itemize}
Given the input and output rates, and according to Equation \eqref{eq:efficiency_general_set1}, the efficiency of $X$ for a given input $I_X$ is bounded by
$$\min(1,I_X\frac{N_L}{U_X},r\frac{N_L}{U_X})-\frac{I_X}{U_X}\text{.}$$
We deduce that the optimal efficiency is bounded by
$$\sup_{I_X \geq 0}\left(\min(1,I_X\frac{N_L}{U_X},r\frac{N_L}{U_X})-\frac{I_X}{U_X}\right)\text{.}$$
If $U_X\leq N_L r$, we get a maximal efficiency $1-\frac{1}{N_L}$ for $I_X=\frac{U_X}{N_L}$, and if $U_X\geq N_L r$, we get $\frac{r(N_L-1)}{U_X}$ for $I_X=r$. Therefore the efficency is never more than $(1-\frac{1}{N_L})\min(1,\frac{N_L r}{U_X})$. This concludes the proof.
\end{proof}

Note that the condition $U_X> N_L r$ corresponds to an overprovisioned system, as the seeders from $X$ have more bandwidth than required to feed the stream $r$ to all leechers by themselves. In the definition of efficiency we proposed, it is normalized by the dedicated upload bandwidth, so overprovisioned systems naturally have lower efficiencies. On the other hand, for any non-overprovisioned system, Equation \eqref{eq:perf_opt} simplifies to

\begin{equation}
\eta(X) = 1-\frac{1}{N_L}\text{.}
\label{eq:eff_ideal}
\end{equation}


In other words, seeders are asymptotically optimal in a perfect P2P live streaming system. The explanation is that the only bandwidth waste boils down to at most one streamrate redirected to them for replication.

\subsection{Limited fanout}
\label{sec:limited}
Each seeder $s$ has now a limited fanout $c_s$. Without loss of generality, we assume that $\forall s\in S, c_s\leq N_L$.
\begin{theorem}
The optimal efficiency of a single seeder $s\in S$ with limited connections $c_s$ is
\begin{equation}
\eta(s)=(1-\frac{1}{c_s})\min(1,\frac{rc_s}{u_s})\text{.}
\label{eq:eff_limited}
\end{equation}
In particular, if $rc_s\geq u_s$ (the fanout is high enough for allowing to use all the upload of $s$), we just have 
\begin{equation}
\eta(s)=1-\frac{1}{c_s}\text{.}
\label{eq:eff_limited_simple}
\end{equation}
\end{theorem}
\begin{proof}
As $s$ cannot reach more than $c_s$ peers, we just consider a sub-system made of $C$, $s$ and $c_s$ leechers, and we conclude by applying Theorem \ref{thm:perfect}, with $c_s$ instead of $N_L$.
\end{proof}

The bad news is that this result stands for a single seeder, and is not easy to extend to a set of seeders. Equation \eqref{eq:efficiency_opt_def} can be a strict inequality, meaning that efficiency is lost in the process of making multiple seeders work together. Consider for instance a toy system made of $N_L=3$ leechers and two seeders $s_1$ and $s_2$ with parameters $u_1=\frac{3}{2}r$, $c_1=2$, $u_2=r$, $c_2=3$. Using Equation \eqref{eq:eff_limited}, we get $\eta_{OPT}(s_1)=\frac{1}{2}$ and $\eta_{OPT}(s_2)=\frac{2}{3}$, so $$\frac{\eta_{OPT}(s_1)u_1+\eta_{OPT}(s_2)u_2}{u_1+u_2}=\frac{17}{30}\text{.}$$

But if we try to find a scheme that maximize the efficiency of $\{s_1,s_2\}$, the best solution leads to
$$\eta_{OPT}({\{s_1,s_2\}})=\frac{8}{15}<\frac{17}{30}\text{.}$$

The good news is that for specific scenarios, we can have $\eta_{OPT}(X)=\frac{\sum_{s\in X}\eta_{OPT}(s)u_s}{U_X}$. This is for instance the case when $X$ is proportionally homogeneous.

\begin{theorem}
\label{thm:eff_limited_homo}
Consider a set $X\subseteq S$ that is proportionally homogeneous, i.e. there is a rate $e$ so that $u_s=ec_s$ for all $s\in X$. Then, for $N_X\leq \lfloor\frac{N_L-1}{\max_{s\in X}(c_s)-1}\rfloor\lfloor \frac{r}{e} \rfloor$
\begin{equation}
\eta_{OPT}(X)=\frac{\sum_{s\in X}(1-\frac{1}{c_s})u_s}{U_X}=\frac{\sum_{s\in X}\eta_{OPT}(s)u_s}{U_X}\text{.}
\label{eq:eff_limited_homo}
\end{equation}
\end{theorem}

Note that although we did not precise $e\leq r$, it is an implicit condition: otherwise, the result only apply for $N_X\leq 0$, or in other words, the empty set.

\begin{corollary}
\label{thm:eff_limited_const}
If all seeders in $X$ have the same upload $u$, maximal fanout $c$, and if $N_X\leq \lfloor\frac{N_L-1}{c-1}\rfloor\lfloor \frac{cr}{u} \rfloor$, then
\begin{equation}
\eta_{OPT}(X)=1-\frac{1}{c}\text{.}
\label{eq:eff_limited_const}
\end{equation}
\end{corollary}

\paragraph*{Remark} In the homogeneous case, if we neglect truncation effects, the condition of Corollary \ref{thm:eff_limited_const} corresponds to $U_X\leq (N_L-1)r \frac{c}{c-1}$. As $(N_L-1)\frac{c}{c-1}\geq N_L$ (because $c\leq N_L$), we get the sufficient condition $U_X\leq r N_L$. Therefore, Corollary \ref{thm:eff_limited_const} can be interpreted as follows: in the homogeneous limited fanout model, up to truncation effects, efficiencies can be aggregated without loss for any non-overprovisionned subset $X$.

\begin{proof}
Given Equations \eqref{eq:efficiency_opt_def} and \eqref{eq:eff_limited_simple}, we just need to give a diffusion scheme $d$ such that  $\eta_{d}(X)=\frac{\sum_{s\in X}\eta_{OPT}(s)u_s}{U_X}$.

That diffusion scheme is the following: the streamrate $r$ is divided into $\lfloor\frac{r}{e}\rfloor$ distinct substreams of rate $e$. We then build up to $\lfloor\frac{r}{e}\rfloor$ trees such that: each seeder $s$ in $X$ is an internal node for exactly one tree, having exactly $c_s=\frac{u_s}{e}$ children; the leaves are taken among the leechers; a leecher can belong to several trees, but is contained at most once per tree.

A given tree can have up to $N_L$ leaves, but no more. We deduce that one tree can contain $\lfloor\frac{N_L-1}{\max_{s\in X}(c_s)-1}\rfloor$ seeders, because a tree with $k$ internal nodes (from $X$) has at most $k(\max_{s\in X}(c_s)-1)+1$ leaves. So the rules of the scheme can be respected if $N_X\leq \lfloor\frac{N_L-1}{c-1}\rfloor\lfloor \frac{cr}{u} \rfloor$.

In the corresponding diffusion scheme, where each tree is used to transmit one of the $\lfloor\frac{r}{e}\rfloor$ distinct substreams of rate $e$, we verify that each seeder $s$ works at optimal efficiency $1-\frac{1}{c_s}$. Equation \eqref{eq:efficiency_general_set} concludes the theorem. The corollary is just a special case where $e=\frac{u}{c}$ and $\max_{s\in X}(c_s)=c$.


\end{proof}

\paragraph*{Remark} We can see in the proof that the bound on $N_X$ is actually related to the numbers of seeders that can fit in a tree with the constraints that each seeder $s$ is an internal node with exactly $c_s$ children and there are no more than $N_L$ leaves. The bound we gave is very conservative, because it assumes $\max_{s\in X}(c_s)$ children for all seeders. It may not be tight, especially if $c_s$ spans a wide range. However, finding out the optimal number of seeders that can collaborate at optimal efficiency is difficult, as it is equivalent to solving a multiple knapsack problem.

\section{Explicit overhead}
\label{subsec:optimal}

From now on, we will focus on the explicit linear overhead model, with proportional cost $a$ and additive cost $b$. Under this model, the bandwidth required for sending one copy of the stream through a single connection is $R:=(1+a)r+b$. One easily checks that $\eta_{\max}:=\frac{r}{R}$ is the maximal efficiency achievable in our model for any peer (leecher or seeder).

When illustrating our results with numerical example, we consider a live streaming system with $r=100$ KBytes/s, a proportional overhead of $10\%$ ($a=0.1$), and two possible additive costs, small ($b=1.7$ KBytes/s) and large ($b=25$ KBytes/s). In the figure, we use the \emph{relative} efficiency $\eta/\eta_{\max}$ instead of $\eta$, in order to facilitate the comparison between the two overhead settings.

\subsection{Efficiency of a single seeder: main theorem}

%

The following theorem gives the optimal efficiency of one single seeder when the overhead is linear.

\begin{theorem}
\label{thm:overhead}
If the overhead follows a linear function, then the optimal efficiency of a seeder $s$ is $\eta(s)=\frac{(N_L-1)r}{u_s}$ if $u_s\geq N_L R$. If $u_s< N_L R$, then we have
\begin{equation}
\eta(s) = \left\{\begin{array}{ll}
0 & \text{ if $0\leq u_s\leq 2b$,}\\
\frac{(1-\sqrt{\frac{b}{u_s}})^2}{1+a} - \epsilon_1(u_s) & \text{ if } 2b\leq u_s\leq \frac{R^2}{b}\text{,}\\
\frac{r}{R}-\frac{r}{u_s} - \epsilon_2(u_s)& \text{ if $u_s\geq \frac{R^2}{b}$, with}
\end{array}\right.
\label{eq:eff_overhead}
\end{equation}
\begin{equation*}
\left\{
\begin{array}{l}
0\leq \epsilon_1(u_s) \leq \frac{1}{1+a}\left( \frac{b}{u_s}\right)^{\frac{3}{2}}\text{,}\\
0\leq \epsilon_2(u_s)\leq \frac{1}{1+a}\frac{b}{u_s} \leq \frac{1}{1+a}\left(\frac{b}{R}\right)^2\text{.}
\end{array}
\right.
\end{equation*}
\end{theorem}

%

\begin{proof}
The easy part of the proof is for $u_s\geq N_L R$. This corresponds to an overprovisioned situation where $s$ alone can provide the live content to all leechers. This is the optimal scheme for $s$, so it is straightforward that $\eta(s)=\frac{(N_L-1)r}{u_s}$.

Equation \eqref{eq:eff_overhead}, which corresponds to the case $u_s< N_L R$, can be proved in three steps:
\begin{itemize}
\item
Finding the maximal efficiency for a given bandwidth $u$ and fanout $c$;
\item
Maximizing the corresponding equations for a continuous $c$;
\item
Bounding the gap induced by the fact that $c$ has to be an integer.
\end{itemize}
\subsubsection{Maximizing $\eta$ for given $u,c$} we first notice that for achieving maximal efficiency, all output rates have to be equal to the input rate: if it is not the case in a given scheme, replacing all output rates by their average value allows to reduce the input rate to that average value (it had to be greater than the maximal output in the original case), increasing efficiency. Therefore the optimal efficiency must be of the form $\eta(s)=\frac{(c-1)e}{u}$, for some rate $0\leq e \leq r$.
It is then obvious that one have interest to choose the highest value of $e$ that is feasible.

Note that if $c=1$, the seeder can only replicate its input and has null efficiency; the seeder needs to maintain at least $2$ connections with spared bandwidth to have a non-null efficiency. This settles that $\eta=0$ for $u\leq 2b$. Otherwise, two cases are to be considered:
\begin{itemize}
\item
if $c$ is the bottleneck (this happens for $u\geq Rc$), then $s$ has enough bandwidth to broadcast the whole stream $r$ to $c$ targets, achieving efficiency $\frac{(c-1)r}{u}$;
\item
if $u$ is the bottleneck (for $u< Rc$), then the optimal input rate $e$ is solution of $c((1+a)e+b)=u$, leading to $e=\frac{\frac{u}{c}-b}{1+a}$. Corresponding efficiency is
\begin{eqnarray*}
\eta&=&\frac{(c-1)e}{u}
=\frac{(c-1)(\frac{u}{c}-b)}{(1+a)u}\\
&=&\frac{1-\frac{1}{c}-\frac{b}{u}(c-1)}{1+a}
\end{eqnarray*}
\end{itemize}

For $u<RN_L$, the bottleneck is necessary one of the above, so we deduce that the optimal efficiency for given $u$ and $c$ is
\begin{equation}
\label{eq:eta_u_c}
\eta(u,c)=\min(\frac{(c-1)r}{u},\frac{(1-\frac{1}{c})-\frac{b}{u}(c-1)}{1+a})
\end{equation}

\subsubsection{Maximizing $\eta$ for given $u$}
We now see \eqref{eq:eta_u_c} as a function of $c$ and try to find its maximal value. We propose to first solve the problem in $\mathbb{R}$ before considering integers.

We introduce
\begin{eqnarray*}
\eta_1(c)& :=& \frac{(c-1)r}{u}\text{ and} \\
\eta_2(c)& :=& \frac{(1-\frac{1}{c})-\frac{b}{u}(c-1)}{1+a}\text{.}
\end{eqnarray*}

The two functions have the following properties:
\begin{itemize}
	\item $\eta_1$ is always increasing, and positive for $c\geq 1$;
	\item $\eta_2$ goes to $-\infty$ for $c$ going to $0$ and $+\infty$. It has a unique maximum $\frac{(1-\sqrt{\frac{b}{u_s}})^2}{1+a}$, which is reached for $c=\sqrt{\frac{u}{b}}$
	\item $\eta_1=\eta_2$ for $c=1$ (corresponding efficiency is $0$) and $c=\frac{u}{R}$ (corresponding efficiency is $\frac{r}{R}-\frac{r}{u}$).
\end{itemize}

We deduce that the optimal efficiency for given $c$, $\eta=\min(\eta_1,\eta_2)$ is equal to $\eta_1$ for $1\leq c \leq \frac{u}{R}$ and $\eta_2$ for $c\geq \frac{u}{R}$. Two cases are then to be considered:
\begin{itemize}
	\item if $\sqrt{\frac{u}{b}}\leq \frac{u}{R}$ (that is $u\geq \frac{R^2}{b}$), then $\eta$ is increasing for $1\leq c \leq \frac{u}{R}$, decreasing for $c \geq \frac{u}{R}$. The maximal efficiency is therefore  $\frac{r}{R}-\frac{r}{u}$, reached for    $c=\frac{u}{R}$;
	\item  if $\sqrt{\frac{u}{b}}\geq \frac{u}{R}$ (that is $u\leq \frac{R^2}{b}$), then the maximal efficiency is the one of $\eta_2$, $\frac{(1-\sqrt{\frac{b}{u_s}})^2}{1+a}$, reached for $c=\sqrt{\frac{u}{b}}$.
\end{itemize}

%
%
\subsubsection{Bounding the quantification gap}

While the optimal value $c_{OPT}$ we found is a real number, only integer value are eligible. However, as the function $\eta=\min(\eta_1,\eta_2)$ always admits a unique maximum, the effective optimal efficiency $\eta(s)$ is necessarily  $\max(\eta(\lfloor c_{OPT} \rfloor),\eta( \lceil c_{OPT} \rceil))$. In particular, we have
$\eta(c_{OPT}+1)\leq \eta(s) \leq \eta(c_{OPT})$, from which we deduce

$$\eta(s)= \eta(c_{OPT}) -\epsilon\text{, with }0\leq \epsilon \leq \eta(c_{OPT}) - \eta(c_{OPT}+1)$$

From there, noticing that $\eta=\eta_2$ for $c\geq c_{OPT}$, we get 

$$\eta(c_{OPT}) - \eta(c_{OPT}+1)=\frac{\frac{b}{u}-\frac{1}{c_{OPT}(c_{OPT}+1)}}{1+a}\text{.}$$

\begin{itemize}
	\item If $u\leq \frac{R^2}{b}$, then $c_{OPT}=\sqrt{\frac{u}{b}}$, so we get
\begin{eqnarray*}
\eta(c_{OPT}) - \eta(c_{OPT}+1)&=& \frac{\frac{b}{u}(1-\frac{1}{1+\sqrt{\frac{b}{u}}})}{(1+a)}\\
&\leq&\frac{(\frac{b}{u})^{\frac{3}{2}}}{(1+a)}\text{;}
\end{eqnarray*}
\item if $u\geq \frac{R^2}{b}$, we just use 
$$\eta(c_{OPT}) - \eta(c_{OPT}+1)\leq \frac{b}{u(1+a)}\text{,}$$
and note that $\frac{b}{u}\leq (\frac{b}{R})^2$. This concludes the proof. 
\end{itemize}

%
%
%
%

\end{proof}

\subsection{Efficiency of a single seeder: discussion}

Following theorem~\ref{thm:overhead} and proof, the following remarks can be made.

\begin{figure}[ht]%
\centering
\includegraphics[width=.96\columnwidth]{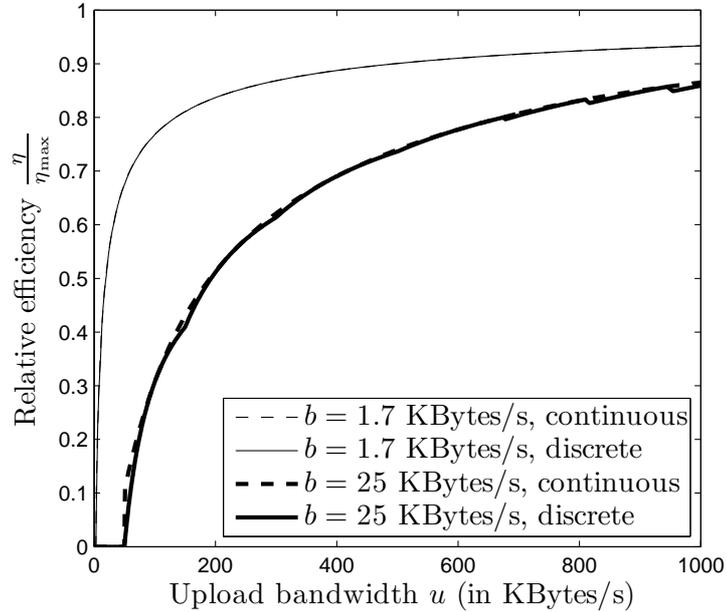} 
\caption{Validity of the continuous approximation of the optimal efficiency}
\label{fig:quantification}%
\end{figure}

\subsubsection{Closed formulas approximation}

the $\epsilon_1$ and $\epsilon_2$ terms are negligible as long as $u_s$ is big enough compared to the additive cost $b$, so in most cases, one can safely use the continuous optimum $\eta(c_{OPT})$ (step \emph{2)} of the proof) instead of the discrete one $\max(\eta(\lfloor c_{OPT} \rfloor),\eta( \lceil c_{OPT} \rceil))$. In other words,
\begin{equation}
\eta(s)\approx \left\{\begin{array}{ll}
\frac{(1-\sqrt{\frac{b}{u_s}})^2}{1+a} & \text{ if } 2b\leq u_s\leq \frac{R^2}{b}\text{,}\\
\eta_{\max}-\frac{r}{u_s} & \text{ if $u_s\geq \frac{R^2}{b}$.}
\end{array}\right.
\label{eq:single_over_approx}
\end{equation}

To illustrate the validity of this approximation, Figure \ref{fig:quantification} compares it to the exact efficiency for the two numerical settings we proposed at the beginning of this Section. We can see that the difference is barely noticeable for a large additive overhead, and invisible for a small one.

%
%

\subsubsection{Low/medium bandwidth}

The case $u_s\leq\frac{R^2}{b}$ can be interpreted as \emph{the upload bandwidth is no more than $\frac{R}{b}$ times the rate $R$}. In most practical situations, one would expect $b\ll R$, so most seeders would probably fall in this case, which corresponds to low, medium and reasonably high bandwidths.

Within this range, it is interesting to note that both the optimal number of connection and corresponding efficiency are independent of $r$. Moreover, one can note that the number of connections, $\sqrt{\frac{u_s}{b}}$, is quite similar to the empirical formula used in the current BitTorrent mainline client, $\sqrt{0.6u}$~\cite{carra08impact}.
This makes us think that the results given here could be adapted to other scenarios than live seeding (this would need to be further investigated in a future work). The $0.6$ factor would corresponds to an additive connection cost $b\approx 1.7$ KBytes/s, which explains why we use this value as one of our numerical settings (the other value, $b=25$  KBytes/s, is totally arbitrary).

\subsubsection{(Very) high bandwidth}

For very high bandwidths (corresponding for instance to seeders managed by some provider), the efficiency tends to $\eta_{\max}$ as $u_s$ goes to infinity (under the assumption that the scenario is not overprovisioned, i.e. $\frac{u_s}{N_L}<R$): super-seeders can asymptotically reach the best achievable efficiency given the overhead constraints.

%
%
%
%
%

\subsubsection{Importance of input shaping}

Seeders do not need to get the whole streamrate. This fact allows to adjust their input rate as desired, which is a key to achieve optimal efficiency.

For instance, under the assumption that the input rate of a seeder $s$ is $r$, one easily checks that its best achievable efficiency is
\begin{equation}
\eta_r(s)=\max(0,\frac{r(\lfloor \frac{u_s}{R} \rfloor-1)}{u_s},\frac{1-\frac{b}{u_s}\lceil\frac{u_s}{R}\rceil}{1+a}-\frac{r}{u_s})
\label{eq:input_r}
\end{equation}
(the case $0$ corresponds to situations where the best choice is not to use $s$, saving the input rate).

\begin{figure}[htb]%
\centering
\includegraphics[width=.96\columnwidth]{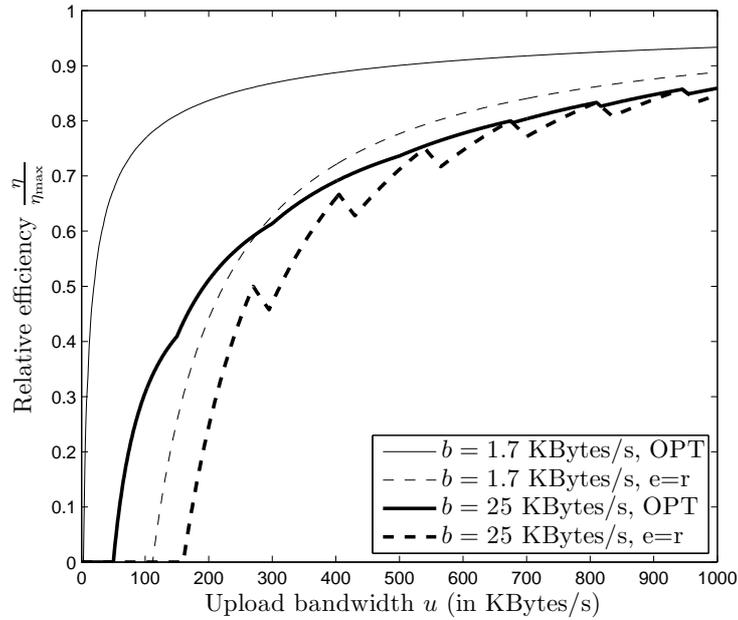} 
\caption{Impact of a badly shaped input rate}
\label{fig:stupid}%
\end{figure}

Figure \ref{fig:stupid} gives a graphical comparison of $\eta_{OPT}$ and $\eta_r$. While seeders with optimized input rates can get a decent efficiency starting from a few $b$'s of upload bandwidth, if the input is $r$, seeders with an upload bandwidth less than $R$ are totally inefficient (they cannot give more than they receive, so the best choice is not to use them). We also notice that the difference remains important even for higher upload bandwidth, especially if the additive overhead is small.

%
%
%


\subsubsection{About receiver-side overhead}
\label{sssec:receiver_overhead}

In our model, we made the assumption that the burden of the overhead was only on the sender. A more general model would consist in assuming that in addition to the sender overhead of parameters $(a,b)$, there is a receiver overhead of parameters $(a_r,b_r)$ (if $p$ receives a streamrate $r_{q,p}$ from $q$, it has to use an upload bandwidth $a_rr_{q,p}+b_r$).

Theorem \ref{thm:overhead} and proof can be adapted to the general model, at the price of increased complexity. For instance, in the medium range scenario ($2b<u_s\leq \frac{R^2}{b}$), we have an optimal (continuous) number of connections

\begin{equation}
c_{OPT}=\sqrt{\frac{u_s}{b}}\text{.}
\label{eq:copt_sender}
\end{equation}

In the general model, this would become 

\begin{equation}
c_{OPT}= \frac{-a_r b + \sqrt{b\left(a + a_r + 1\right) \left(u - b_d - a b_r + a_r b + a u\right)}}{b \left(a + 1\right)}
\text{.}
\label{eq:copt_full}
\end{equation}

\begin{figure}%
\centering
\includegraphics[width=.96\columnwidth]{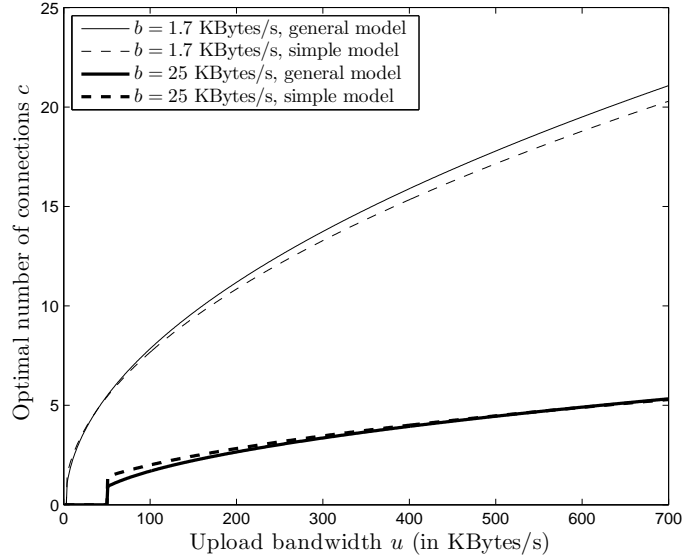} 
\caption{General overhead model vs simple overhead model}
\label{fig:full_overhead}%
\end{figure}

We see that formulas get much more complex in the general model. However, if one compares the practical values given by \eqref{eq:copt_sender} and \eqref{eq:copt_full}, we see that the general behavior remains practically the same. This is depicted in Figure \ref{fig:full_overhead} (receiver overhead is assumed to be the same that the sender overhead, i.e. $a_r:=a$ and $b_r:=b$).

As the added complexity does not seem to bring lot of practical difference, we choose to discard the receiver overhead in our model. However, the reason we can do that is probably that the natural use of live seeders is to feed them with a single input rate, which reduce the impact of receiver overhead. If we want to extend our framework to leechers, which usually receive multiple substreams from multiple sources, a proper modeling of the receiver overhead may become mandatory.


%
%
%
%
%
%
%
%


\subsection{Efficiency of a set of seeders}
\label{sec:efficiency}

Like for the limited fanout model, there is no guarantee that the optimal single efficiencies of seeder can be aggregated in a common scheme. In the following, we propose two heuristics that allow to somehow adapt Theorem \ref{thm:eff_limited_homo} to the overhead model: the mono-rate and dichotomic rates diffusion schemes.


\subsubsection{Mono-rate scheme}

The idea of the mono-rate approach is somehow simple: if a set of seeders agree to a common substream rate $e$, they can behave as a proportionally heterogeneous set. Their efficiency obeys to the following theorem:

\begin{theorem}
\label{thm:mono}
Consider a set $X\subseteq S$ that verifies:
\begin{itemize}
	\item $\bar{u}_X\leq \frac{2R^2}{b}$;
	\item $N_X\leq\lfloor \frac{N_L-1}{\lfloor \frac{\max_{s\in X}(u_s)}{E} \rfloor-1}\rfloor \lfloor \frac{R-b}{E-b} \rfloor$, with $E=\sqrt{\frac{b\bar{u}_X}{2}}$.
\end{itemize}
Then, if all seeders on $X$ agree on a common rate $e:=\frac{E-b}{1+a}$ used for all inputs and outputs, the efficiency $\eta_e(X)$ of the corresponding scheme verifies
\begin{equation}
\frac{(1-\sqrt{\frac{2b}{\bar{u}_X}})^2}{1+a}<\eta_e(X)\leq\frac{(1-\sqrt{\frac{b}{\bar{u}_X}})^2}{1+a}
\label{eq:mono}
\end{equation}
\end{theorem}

\begin{proof}
Consider a given rate $e\leq r$. Call $E:=(1+a)e+b$ the corresponding rate with overhead. The maximal efficiency of a seeder $s$ having $e$ as input and ouputs is reached when $s$ opens the maximal number of outgoing connections allowing to stream $e$. This leads to
$$\eta_e(s)=\frac{(\lfloor\frac{u_s}{E}\rfloor-1)e}{u_s}\text{.}$$
In particular, 
\begin{equation*}
\frac{e}{E}-2\frac{e}{u_s}<\eta_e(s)\leq\frac{e}{E}-\frac{e}{u_s}\text{.}
\end{equation*}
Assume that the number of seeders in $X$ is small enough to allow perfect aggregation of efficiencies, like for Theorem \ref{thm:eff_limited_homo} (the corresponding condition will be derived later). We then have $\eta_e(X)=\frac{\sum_{s\in X}\eta_e(s)u_s}{U_X}$, therefore
\begin{equation*}
\frac{e}{E}-2\frac{e}{\bar{u}_X}<\eta_e(X)\leq\frac{e}{E}-\frac{e}{\bar{u}_X}\text{.}
\end{equation*}
The maximal value of $\frac{e}{E}-\frac{e}{\bar{u}_X}$ is $\frac{(1-\sqrt{\frac{b}{\bar{u}_X}})^2}{1+a}$, proving the right part of \eqref{eq:mono}.
The maximal value of $\frac{e}{E}-2\frac{e}{\bar{u}_X}$ is $\frac{(1-\sqrt{\frac{2b}{\bar{u}_X}})^2}{1+a}$, and it is reached for $E=\sqrt{\frac{b\bar{u}_X}{2}}$. As we have $E\leq R$, this implies $\bar{u}_X\leq \frac{2R^2}{b}$.

We then need to give a sufficient condition for aggregating the efficiencies without losses. We can use the condition from Theorem \ref{thm:eff_limited_homo}, $N_X\leq \lfloor\frac{N_L-1}{\max_{s\in X}(c_s)-1}\rfloor\lfloor \frac{r}{e} \rfloor$. Noticing that $c_s=\lfloor \frac{u_s}{E}\rfloor$ allows to conclude.
\end{proof}

\subsubsection{Dichotomic scheme}

The dichotomic approach consists in the diffusion of several substreams whose rates are dividers of $r$, instead of using a single rate $e$. In details, the predetermined substreams are:
\begin{itemize}
	\item The video stream of rate $r$, which can be split into
	\item $2$ non-overlapping substreams of rate $\frac{r}{2}$, each of which can be split into 2 substreams
	\item \ldots
	\item $2^{k_{\max}}$ non-overlapping substreams of rate $\frac{r}{2^{k_{\max}}}$, for some $k_{\max}\geq 0$.
\end{itemize}

$k$ is called the level of a substream of rate $\frac{r}{2^{k}}$

A seeder $s$ is said to operate at level $k$ if it behaves as follows:
\begin{itemize}
	\item it receives as input a level $k$ substream; let $l:=k$ be his working level;
	\item As long as $s$ has a residual upload bandwidth greater than $b$ and $l\leq k_{\max}$, do:
	\begin{itemize}
	\item if there is not enough residual upload bandwidth to establish a new output of level $l$,
	\item then $l=l+1$ (a children substream of the current level $l$ substream is chosen),
	\item else create a new output of level $l$.
\end{itemize}
\end{itemize}

\begin{figure}%
\centering
\includegraphics[width=.96\columnwidth]{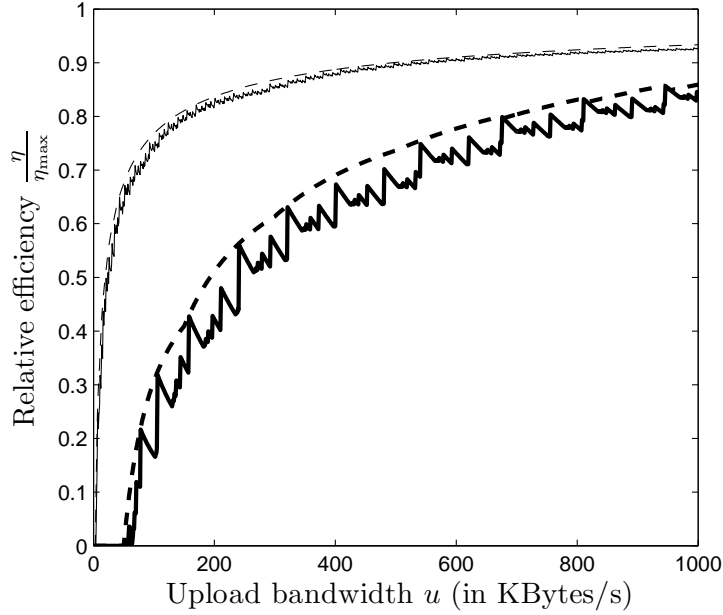} 
\caption{Dichotomic vs optimal individual efficiencies}
\label{fig:bin}%
\end{figure}

The corresponding efficiency is denoted $\eta_k(s)$. In order to optimize the dichotomic approach, each seeder operates at a level that maximizes its single efficiency, i.e. chooses a level $k_s$ such that $\eta_{k_s}(s)=\max_{0\leq k\leq k_{\max}}\eta_k(s)$. The corresponding  efficiency is denoted $\eta_{Bin}(s)$.

As the operating rate is necessarily a divider of $r$, $\eta_{Bin}(s)$ is necessarily suboptimal. However, the different levels allow enough freedom to get an efficiency close enough to be optimal. For instance, Figure \ref{fig:bin} gives a graphical comparison of $\eta_{Bin}(s)$ and $\eta_{OPT}(s)$, using $k_{\max}=\lfloor \log_2(\frac{r}{b}) \rfloor$ (this is an arbitrary choice that corresponds to stopping the subdivision when substreams need more overhead that their actual goodput). One observes that the individual efficiency loss is quite sustainable, especially for a low additive overhead.

For a given set $X$ of seeders, the construction of a dichotomic diffusion scheme is rather simple:
\begin{itemize}
	\item all seeders operating at level $k$ organize to achieve up to $2^k$ diffusions tree for the level $k$; each seeder try to join the level $k$ diffusion tree which currently possesses less leaves.
	\item if a level $k$ seeder has outputs of level $k'>k$, they can either be directly transmitted to leechers or serve as root for a level $k'$ diffusion tree;
	\item if some seeders at level $k$ miss the input streamrate to build their diffusion scheme, they may use a leaf from a parent substream diffusion tree (some of parent rate will be wasted).
\end{itemize}

%

Under some conditions, we can evaluate the efficiency of $X$ under a dichotomic diffusion.

\begin{theorem}
\label{thm:dicho}
If, for a given set $X\subseteq S$, we have $U_X\leq N_L R$, and if all non-empty diffusion trees can be rooted with proper input, then the efficiency $\eta_{Bin}(X)$ of $X$ under a dichotomic diffusion verifies
\begin{equation}
\frac{\sum_{s\in X}\eta_{Bin}(s)u_s}{U_X} - \frac{rk_{\max}}{U_X}\leq \eta_{Bin}(X)\leq \frac{\sum_{s\in X}\eta_{Bin}(s)u_s}{U_X}
\label{eq:dicho}
\end{equation}
\end{theorem}

The interpretation is the following: up to a term $\frac{rk_{\max}}{U_X}$, which is small if $U_X$ is big enough, the individual dichotomic efficiencies, which are close to the optimal individual efficiencies, can be aggregated without loss.

\begin{proof}
The condition $U_X\leq N_L R$ ensures that no diffusion tree has more leaves than there are leechers in need of the corresponding substream. This can be shown by induction:
\begin{itemize}
	\item at level $0$, the diffusion tree cannot have more than $\lfloor \frac{U_X}{R}\rfloor$ leaves, which is smaller than $N_L$.
	\item at level $k$, let $U_{k}$ denote the bandwidth that remains after the bandwidth consumed from lower level is substracted; let $N_k$ the maximal number of leechers that can be leaves at that level (a given leecher is counted with multiplicity equal to the number the level $k$ substream it needs; let $M_k$ the number of leechers that get a level $k$ substream (with multiplicity). Note the relation $N_k=2(N_{k-1}-M_{k-1})$, i.e. the maximal number at a given level is twice the slots that have not been filled in the previous level. Assume that $U_{k-1}\leq N_{k-1}((1+a)\frac{r}{2^{k-1}}+b)$, that is at level $k-1$, the residual bandwidth is not overprovisioned compared to the number of possible leaves
Then we have

	\begin{eqnarray*}
	U_k&\leq & U_{k-1}-M_{k-1}((1+a)\frac{r}{2^{k-1}}+b)\\
	& \leq & (N_{k-1}-M_{k-1}) ((1+a)\frac{r}{2^{k-1}}+b)\\
	& \leq & N_k ((1+a)\frac{r}{2^{k}}+\frac{b}{2}) \leq N_k ((1+a)\frac{r}{2^{k}}+b)
\end{eqnarray*}

\end{itemize}

So at any given level, a diffusion tree can always find a leecher to give its output to. Therefore the only waste compared with individual efficiencies lies when the root input of a tree comes from a parent substream. This is bounded by $r$ when considering all roots at a given level $k>0$, leading to a total waste bounded by $rk_{\max}$. Normalizing by $U_X$ concludes the proof.
\end{proof}

%
%
%
%
%

\subsubsection{Comparison of the two methods}

The mono-rate approach is simple to describe, which makes it a good proof of concept of using multiple seeders in a system with overhead.
However, the dichotomic approach, although more complex, has many advantages over the mono-rate approach that make it more suitable for a practical use.

Firstly, the substreams are pre-determined, while mono-rate requires to determine the proper input rate $e$, which depends on $\bar{u}_X$. Among other things, this facilitate considerably the interaction with the leechers' diffusion process. 
 Furthermore, under the dichotomic approach, a seeder $s$ can determine its operating level by itself (it is just a function of $u_s$) while in the mono-rate approach, knowing $\bar{u}_X$ implies some knowledge of the whole set $X$. This is even worse when considering dynamics in $X$: A change in $e=f(\bar{u}_X)$ requires a complete upset of the diffusion trees in the mono-rate approach, while changes are expected to be mostly local in the dichotomic approach. 

Also note that as streamrate are dividers of $r$, the quantification effect $\lfloor \frac{r}{e} \rfloor$ that may limit the mono-rate approach (cf Theorem \ref{thm:mono}) has no equivalent in the dichotomic approach.

Finally, the mono-rate approach can force lot of seeders to use an input rate that is far from the single seeder optimal. This impact is bounded (cf Theorem \ref{thm:mono}), but can be non negligible, especially if the seeders' bandwidths are highly heterogeneous. In contrast, the dichotomic approach adjusts afor each seeder $s$ a level $k_s$ such that the input rate is to far from the optimal.

\section{Discussion}
\label{sec:discussion}

\subsection{Leecher diffusion process}

We did not consider in details the way to make  the diffusion processes of leechers and seeders work together. This is a problem in itself, which deserves a separate study. 
The study performed in \cite{liu08performance} seems to be adaptable to the case with seeders, at least for the limited fanout model, but a further work is required to transpose the results to the overhead model (including keeping in mind the existence of receiver-side overhead).

However, we argue that knowing how to optimize the diffusion process of seeders alone is not a bad starting point.

\subsection{Make a minimal use of seeders}

While all this paper is devoted to make the best possible use of seeders, we should recall that in the design of a real system, targeting the maximal seeder efficiency is not necessarily the smartest thing to do.

In fact, seeders ``waste'' their input rate by design, which makes them inherently less efficient that leechers. Therefore, one should use seeders as minimally as possible. The proper way to use seeders is:
\begin{itemize}
	\item Try to achieve the most of the content diffusion by using the servers and leechers alone. If possible, the leechers should perform a lossless diffusion of a common substream of rate $r'\leq r$ among all of them instead of a partial or lossy diffusion of rate $r$;
	\item if $r'<r$, use seeders to finish the job. This is were the results of this paper apply, which describe the best one can expect from seeders and how to achieve it.
\end{itemize}

%


\subsection{Application: dimensioning a scalable live streaming system}
\label{sec:application}
%
%
%


Many dimensioning rules can be derived by using the formulas we proposed. For instance, determining if the system is scalable would consist in checking if $\eta(L)\alpha_L+\eta(S)\beta\alpha_S\geq 1$ \cite{benbadis08playing}. If we assume here for simplicity homogeneous bandwidth $u$, $\eta(S)=\eta_{OPT}(u)$ (neglecting aggregation issues), and optimal leechers' efficiency $\eta_L=\eta_{\max}$\footnote{The efficiency of leechers should take into account the number of outgoing connections like we did for the seeders. However, $\eta_L$ is not the main matter of this paper, so we assume without remorse perfect efficiency $\eta_{\max}$.}, one can derive the relationship that $u$ and $\beta$ must verify for the system to be scalable:

\begin{equation}
\beta \eta_{OPT}(u)\geq \frac{r}{u}-\eta_{\max}\text{.}
\label{eq:u_beta}
\end{equation}

If $\beta$, which indicates the ratio between idle (seeders) and active (leechers) users, is a given parameter of the system, Equation \eqref{eq:u_beta} can be used to derive the bandwidth $u$ that is required for the system to be scalable. This is illustrated by Figure \ref{fig:target} (the performance of the perfect system, i.e. $a=b=0$, is also plotted for comparison). Notice how even little values of $\beta$ (less than $1$) can give significant decrease of the required bandwidth, which is $R$ for a seedless system with perfectly efficient leechers.

\begin{figure}[!t]%
\centering
\includegraphics[width=.96\columnwidth]{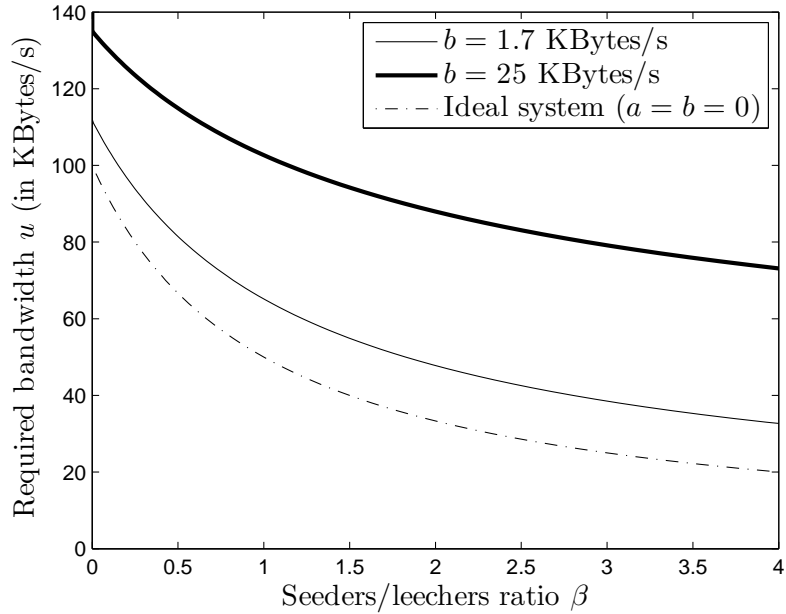} 
\caption{Average bandwidth required for scalability}
\label{fig:target}%
\end{figure}

%
%

\subsection{About delays}

We do not have taken delay issues into account. The diffusion delay is obviously a major concern in the design of a live streaming system. However, it should be noted that the two heuristics we proposed are based on diffusion trees. Therefore the induced delay is at most equal to the delay of a single connection times a logarithm of $N_L$. This is exactly the same type of delay that is experienced for diffusion based on leechers only, so we argue that using seeders should not impact the delay performance of a P2P live streaming system.

\section{Conclusion}
\label{sec:conclusion}

In this paper, we gave the keys to understand how seeders could be used in P2P live streaming if servers and leechers do not suffice. After a preliminary work on perfect and limited-fanout systems, we conducted our study on a model with linear overhead. Although this is a preliminary study, with results that are more theoretical than practical, we believe that the present work may have a significant impact in the design and dimensioning of live streaming systems using seeders.

In a future work, we plan to pursue the matter of leechers/seeders interaction in the general overhead model. We also think that the concept of live seeders introduced here could be extended to a more general concept of half-seeders, i.e. seeders with not all resources expected from a traditional seeder. Studying half-seeders could allow to extend our results to all P2P content distribution systems, including file-sharing and Video-on-Demand systems.

\bibliographystyle{plain}
\bibliography{liveseed}

\end{document}